\documentclass[a4paper,runningheads]{llncs}

\newcommand{\prob}{\operatorname{Prob}}
\newcommand{\poly}{\operatorname{poly}}

\newcommand{\est}[1]{\tilde{#1}}

\usepackage[noend]{algorithmic}
\usepackage[boxed]{algorithm}

\newcommand{\OHS}{\mathcal{O}^*}
\newcommand{\OH}{\mathcal{O}}
\newcommand{\NP}{NP}

\usepackage{amsmath}
\usepackage{amsfonts}
\usepackage{amssymb}
\usepackage{paralist}
\usepackage{graphicx}
\usepackage{xspace}

\newcommand{\subsetsum}{\textsc{Subset Sum}\xspace}
\newcommand{\knapsack}{\textsc{Knapsack}\xspace}
\newcommand{\vc}{\textsc{Vertex Cover}\xspace}
\newcommand{\ds}{\textsc{Dominating Set}\xspace}
\newcommand{\tsp}{\textsc{Traveling Salesman}\xspace}

\title{Reducing a Target Interval to a Few Exact Queries\thanks{This work is supported by the Nederlandse Organisatie voor Wetenschappelijk Onderzoek (NWO), project: 'Space and Time Efficient Structural Improvements of Dynamic Programming Algorithms', and by ERC StG project PAAl no.\ 259515.}}
\author{Jesper Nederlof \inst{1} \and Erik Jan van Leeuwen\inst{2} \and Ruben van der Zwaan \inst{3}
\institute{Utrecht University, The Netherlands, \email{J.Nederlof@uu.nl}.\and Sapienza University of Rome, Italy, \email{E.J.van.Leeuwen@dis.uniroma1.it}. \and Maastricht University, The Netherlands,
\email{r.vanderzwaan@maastrichtuniversity.nl}.} }

\begin{document}
\maketitle
\addtocounter{footnote}{-3}
\begin{abstract}
Many combinatorial problems involving weights can be formulated as a so-called \emph{ranged problem}. That is, their input consists of a universe $U$, a (succinctly-represented) set family $\mathcal{F} \subseteq 2^{U}$, a weight function $\omega:U \rightarrow \{1,\ldots,N\}$, and integers $0 \leq l \leq u \leq \infty$. Then the problem is to decide whether there is an $X \in \mathcal{F}$ such that $l \leq \sum_{e \in X}\omega(e) \leq u$. Well-known examples of such problems include \textsc{Knapsack}, \textsc{Subset Sum}, \textsc{Maximum Matching}, and \textsc{Traveling Salesman}. In this paper, we develop a generic method to transform a ranged problem into an \emph{exact problem} (i.e.~a ranged problem for which $l=u$). We show that our method has several intriguing applications in exact exponential algorithms and parameterized complexity, namely:
\begin{itemize}
\item In exact exponential algorithms, we present new insight into whether \subsetsum and \knapsack have efficient algorithms in both time and space. In particular, we show that the time and space complexity of \subsetsum and \knapsack are equivalent up to a small polynomial factor in the input size. We also give an algorithm that solves sparse instances of \knapsack efficiently in terms of space and time.
\item In parameterized complexity, we present the first kernelization results on weighted variants of several well-known problems. In particular, we show that weighted variants of \vc and \ds, \tsp, and \knapsack all admit polynomial randomized Turing kernels when parameterized by $|U|$.
\end{itemize}
Curiously, our method relies on a technique more commonly found in approximation algorithms.
\end{abstract}

\section{Introduction}\label{sec:intro}

In many computational problems in the field of combinatorial optimization the input partly consists of a set of integers. Since integers are naturally represented in binary, they can be exponential in the number of bits of the input instance. For many problems this is not an issue, particularly for problems admitting a strongly polynomial-time algorithm (recall that the running time of such an algorithm does not depend on the size of the integers).
However, (exponentially) large numbers present a major issue for other problems. For example, in weakly \NP-complete problems the large integers are even the sole source of hardness. Strongly \NP-complete problems are often studied in their weighted variants, and often such weighted variants are even considerably harder than their unweighted counterpart.
In this paper, we present a novel method to reduce the challenges posed by (exponentially) large numbers in the input of \NP-complete problems.

We first give a description of the type of problems that we consider. All of the studied problems can be stated according to the following generic pattern. First, there is a universe $U$ (for example the set of vertices or edges of a graph) and a weight function $\omega:U \rightarrow \{1,\ldots,N\}$. Second, there is a succinctly-represented set family $\mathcal{F} \subseteq 2^U$. We will assume that membership of the set $\mathcal{F}$ can be determined in polynomial time by an oracle given as part of the input. Finally, we are given two non-negative integers $l,u$ such that $0 \leq l \leq u \leq \infty$. Then the problem is to decide whether there exists an $X \in \mathcal{F}$ such that $\omega(X) \in [l,u]$, where $\omega(X)=\sum_{e \in X}\omega(e)$. We call this a \emph{ranged problem}. If a problem additionally specifies that $l=u$, this is an \emph{exact problem}. We will be mainly interested in the case where $N$ is exponential, or even super-exponential, in $|U|$.

The main question that we consider and answer in this paper is whether the computational complexity of a ranged problem is equal to that of its corresponding exact problem. This question is motivated by the recent availability of powerful tools for exact problems, such as hashing (see e.g.~\cite{DBLP:journals/siamcomp/HarnikN10}) and interpolation (see e.g.~\cite{DBLP:conf/stoc/LokshtanovN10,DBLP:journals/siamcomp/Mansour95}), that do not seem directly applicable to ranged problems. Hence we may wonder whether there is a difference between ranged and exact problems from the point of view of computational complexity.

Certain cases of this main question are particularly intriguing. For example, the arguably most fundamental pair of an exact and its corresponding ranged problem is \subsetsum and \knapsack respectively\footnote{In the field of cryptography, ``Knapsack'' is often used to refer to ``Subset Sum''.}. Recall that in \subsetsum, we are given a set $U=\{1,\ldots,n\}$, a weight function $\omega : U \rightarrow \{1,\ldots,N\}$, and an integer $t \leq N$, and we are asked to decide whether there exists an $X \subseteq U$ such that $\omega(X)=t$. In the \knapsack problem we are additionally given a weight function $\nu$ and integer $b$, and we are asked to decide whether there exists a set $X$ with $\omega(X) \geq t$ among all $X \subseteq U$ for which $\nu(X) \leq b$. From the perspective of exact exponential algorithms (see e.g.~\cite{Fominbook,DBLP:journals/dam/Woeginger08} for an introduction), both problems are known to be solvable in $\OHS(2^{n/2})$ time and $\OHS(2^{n/4})$ space \cite{DBLP:journals/siamcomp/SchroeppelS81} (see also \cite[Chapter 9]{Fominbook}), while the best polynomial-space algorithms are still the trivial brute-force $\OHS(2^n)$-time algorithms. It is an interesting question whether either of these problems can be solved in $\OHS(1.99^n)$ time and polynomial space. Also, are the problems related in the sense that an improved algorithm for \subsetsum would imply an improved algorithm for \knapsack?

Another interesting perspective is that of \emph{sparse instances}. It is known that the \subsetsum and \knapsack problems can be solved in pseudo-polynomial time and space using a dynamic programming (DP) algorithm~\cite{Bellman03}. An intensively studied case of DP is where the DP-table is guaranteed to be sparse. In \subsetsum, for example, this means that the number of distinct sums of the subsets of the given integers is small. Using memorization, this type of sparseness can be easily exploited if we are allowed to use exponential space. Very recently, polynomial-space equivalents of memorization were given in~\cite{Kaski11} (see also \cite[Chapter 6]{thesisJesper}). The first step in this approach uses hashing, and the second step uses interpolation. It is unclear whether the approach can be extended to \knapsack. One issue is that a good hash function does not hash the target interval to a single interval. Furthermore, interpolation does not apply directly to the ranged case, and the typical solution of adding a few ``slack weights'' to reduce it to the exact case destroys the sparseness property.

We note that different measures of sparseness for \knapsack have been considered previously. Nemhauser and \"{U}llmann~\cite{Nemhauser1969} considered the case when the number of Pareto-optimal solutions is small. A solution $X$ for \knapsack is Pareto-optimal if there is no $X'$ with $\nu(X') < \nu(X)$ and $\omega(X') \geq \omega(X)$, or with $\nu(X') \leq \nu(X)$ and $\omega(X') > \omega(X)$. Note that the number of Pareto-optimal solutions is always at most the number of distinct sums in the instance. The algorithm of Nemhauser and \"{U}llmann uses $\OH(\sum_{i=k}^n p_i)$ time and $\OH(\max p_i)$ space to enumerate all Pareto-optimal solutions, where $p_{i}$ is the number of Pareto-optimal solutions over the first $i$ items. Note that the space requirement is polynomial in the sparseness, whereas the space requirement of the algorithm of~\cite{Kaski11} is polynomial in the size of the instance. In the framework of smoothed analysis\footnote{Smoothed analysis aims to provide a middle ground between average-case and worst-case analysis. See e.g.~\cite{SpielmanT09}.}, however, the number of Pareto-optimal solutions for \knapsack is polynomial in the instance size~\cite{BeierV04}.

In the field of kernelization (see~\cite{DBLP:conf/iwpec/Bodlaender09} for a survey), the \subsetsum problem is known to admit a so-called \emph{polynomial randomized kernel} when parameterized by the number of integers~\cite{DBLP:journals/siamcomp/HarnikN10}. Can a similar kernel be obtained for the \knapsack problem? Again, since~\cite{DBLP:journals/siamcomp/HarnikN10} heavily relies on hashing, it does not seem to be applicable. Similar questions can be asked for weighted variants of several fundamental problems in the field of kernelization, such as the weighted variant of \vc. Is there a (randomized Turing) kernel for this problem parameterized by $|U|$, i.e.\ can we reduce the weights to be at most $2^{|U|^{O(1)}}$?

\medskip\noindent
{\it Our Results}\quad
In this paper, we show that a ranged problem is equally hard as its corresponding exact problem, modulo a factor $\OH(|U| \cdot \lg (|U|\, N))$ in the running time. This implies a positive answer on all of the above questions.  This result uses a generic and clean method to transform a ranged problem into instances of its corresponding exact problem. The method covers the interval $[l,u]$ with a small number of ``fuzzy intervals'' such that an integer is in $[l,u]$ if and only if it is in one of the fuzzy intervals. It relies on a scaling technique that is more commonly found in approximation algorithms; in fact, the prime example of its use is in the FPTAS for \knapsack~\cite{IbarraK1975}.

The paper is organized as follows. In Section~\ref{sec:prelim}, we introduce the required notation and definitions. In Section~\ref{sec:maintool}, we state and prove our main technical contribution. Sections~\ref{sec:exactalg} and~\ref{sec:kernel} are dedicated to corollaries of the main theorem in the fields of exact exponential algorithms and kernelization. Finally, we give a conclusion, further remarks, and open questions in Section~\ref{sec:conclusion}. 

\section{Preliminaries} \label{sec:prelim}
Throughout, we use the $\OHS(\cdot)$ notation that suppresses any factor polynomial in the input size of the given problem instance. We use Greek symbols such as $\omega$ to denote weight functions, i.e.~for a universe $U$ and an integer $N$, $\omega:U \rightarrow \{1,\ldots,N\}$. In this context, we shorthand $\omega(X)=\sum_{e \in X}\omega(e)$ for any $X \subseteq U$. For two integers $l\leq u$, the set of integers $\{l,l+1,\ldots,u\}$ is denoted by $[l,u]$.

A \emph{kernelization algorithm} (or \emph{kernel}) for a parameterized problem $\Pi$ (that is, a problem together with an input measure $k$) computes in polynomial time, given an instance $(x,k)$ of $\Pi$, a new instance $(x',k')$ of $\Pi$ such that $(x',k') \in \Pi$ if and only if $(x,k) \in \Pi$, and $|x'| \leq f(k)$ for some computable function $f$. The instance $(x',k')$ is called a \emph{kernel} of $\Pi$, and it is called a \emph{polynomial kernel} if $f$ is a polynomial. Not every problem admits a polynomial kernel, or the polynomial hierarchy collapses to the third level~\cite{DBLP:journals/jcss/BodlaenderDFH09}. We refer to~\cite{DBLP:conf/iwpec/Bodlaender09} for a recent overview.

A generalization of the notion of a kernel is a \emph{Turing kernel}. Here the requirement that given a kernel yields an equivalent instance is relaxed. Instead, a polynomial number of instances with the same size restrictions as before may be produced. Moreover, there should be a polynomial-time algorithm that, given which of the produced instances are a {\tt Yes}-instance, decides whether the original instance is a {\tt Yes}-instance. The special case where the algorithm returns the OR of the produced instances is called an OR-kernel or a many-to-one kernel \cite{DBLP:conf/stacs/FernauFLRSV09}. Generalizing these notions further, we use the adjective ``randomized'' to indicate that the polynomial-time algorithm computing the final answer may have a constant one-sided error probability. Interestingly, the results of~\cite{DBLP:journals/jcss/BodlaenderDFH09} even apply to the randomized variant of the original kernel definition, but not to Turing kernels.

Given a graph $G=(U,E)$, a subset $X \subseteq U$ is a \emph{vertex cover} if  $u \in X$ or $v \in X$ for every $(u,v) \in E$, and it is a \emph{dominating set} if $u \in X$ or $(u,v) \in E$ for some $v \in X$ for every $u \in U$. In the weighted \vc and \ds problems, we are given a graph on vertex set $U$ together with a weight function $\omega: U \rightarrow \{1,\ldots,N\}$ and an integer $t$, and are asked to decide whether there is a vertex cover or dominating set $X \subseteq U$, respectively, such that $\omega(X) \leq t$.
A \emph{Hamiltonian cycle} is a subset $X \subseteq E$ such that the graph $(V,X)$ is a cycle. In the \tsp problem we are given a graph on edge set $U$ together with a weight function $\omega: U \rightarrow \{1,\ldots,N\}$ and an integer $t$, and are asked to decide whether there exists a Hamiltonian cycle $X \subseteq U$ such that $\omega(X) \leq t$.

\section{The Main Method} \label{sec:maintool}

\newcommand{\low}{l}
\newcommand{\upp}{u}

In this section, we give the main technical contribution of the paper: we transform a ranged problem into a small number of exact problems. The naive way to obtain such a transformation would be to return a new instance for each $x \in [\low,\upp]$, thus yielding $u-l$ problems. However, as $\upp-\low$ can be exponential in the size of the input, this procedure is clearly not efficient. Instead, Theorem~\ref{thm:shrinkint} develops a new family of $\OH(|U| \lg (|U|\, N))$ weight functions.

We want to stress that the  main conceptual consequence of Theorem~\ref{thm:shrinkint} is that, modulo a small polynomial factor, weighted subset-selection problems that aim to find a subset of given exact weight are equally hard as those that aim to find a subset with weight in a given interval.

\begin{theorem}[Shrinking intervals]\label{thm:shrinkint}
Let $U$ be a set of cardinality $n$, let $\omega: U \rightarrow \{0,\ldots,N\}$ be a weight function, and let $\low<\upp$ be non-negative integers with $u-l>1$. Then there is a polynomial-time algorithm that returns a set of pairs $\Omega=\{(\omega_1,t_1),\ldots,(\omega_K,t_K)\}$ with $\omega_i: U \rightarrow \{0,\ldots, N\}$ and integers $t_1,\ldots,t_K \leq N$ such that
\begin{enumerate}[(C1)]
	\item $K$ is at most $(5n +2) \lg (\upp-\low)$, and
	\item for every set $X \subseteq U$ it holds that $\omega(X) \in [\low,\upp]$ if and only if there exist an index $i$ such that $\omega_i(X) = t_i$.
\end{enumerate}
\end{theorem}
\begin{proof}
The theorem is implemented in Algorithm~\ref{alg:shrink}.
\begin{algorithm}
\caption{Shrinking intervals}
\label{alg:shrink}
\begin{algorithmic}[1] 
\REQUIRE $\mathtt{shrink}(\omega,\low,\upp)$
	\IF{$u-l \leq 5n$}
		\RETURN $\{(\omega,\low),(\omega,\low+1),\ldots,(\omega,\upp)\}$.
	\ELSIF{$l$ is odd}
		\RETURN $\{(\omega,l)\} \cup \mathtt{shrink}(\omega,\low+1,\upp)$.
	\ELSIF{$u$ is odd}
		\RETURN $\{(\omega,u)\} \cup \mathtt{shrink}(\omega,\low,\upp-1)$.
	\ELSE
		\STATE For every $e \in U$, set $\omega'(e)=\left\lfloor \omega(e)/2 \right\rfloor$.
		\STATE $\Omega_l \leftarrow \{(\omega,\low),(\omega,\low+1),\ldots,(\omega,\low+3n)\}$.
		\STATE $\Omega_r \leftarrow \{(\omega,\upp),(\omega,\upp-1),\ldots,(\omega,\upp-2n)\}$.
		\RETURN $\Omega_l \cup \mathtt{shrink}(\omega',(\low+2n)/2, (\upp-2n)/2) \cup  \Omega_r$.
	\ENDIF
\end{algorithmic}
\end{algorithm}
Let $T(u-l)$ denote the maximum number of pairs that $\mathtt{shrink}$ returns. We claim that $T(u-l)$ is at most $(5n+2) \lg (\upp-\low)$. 
The cases when either $\low$ or $\upp$ is odd can only happen twice in a row before either case one or four occurs. Observe that for the first case our claim is clearly true, and that for the last case we obtain the bound $T(u-l) \leq (5n+2) + T((u-l)/2)$, which clearly meets our claim since $u-l \geq 2$. This settles (C1) and the claim concerning the running time.

We prove (C2) by induction. For the thirst three cases, (C2) clearly holds, so let us directly proceed to the last case. For the forward direction, assume that $\omega(X) \in [\low,\upp]$. If $\omega(X) \in [\low,\low+3n] \cup [\upp -2n, \upp]$, then a pair fulfilling (C2) is in $\Omega_l \cup \Omega_r$. So assume that $\omega(X) \in [\low+3n,\upp-2n]$. Then a pair fulfilling (C2) will be added in the recursive step by the induction hypothesis, because $\omega'(X) \in [(\low+2n)/2,(\upp-2n)/2]$ since
\begin{equation}\label{eq:scale}
\frac{\omega(X) - n}{2}\leq \sum_{e \in X} \left\lfloor \frac{\omega(e)}{2} \right\rfloor = \omega'(X) \leq \omega(X)/2.
\end{equation}
For the reverse direction, assume that there is a pair $(\omega_i,t_i) \in \Omega$ such that $\omega_i(X)=t_i$. If the pair is from $\Omega_l \cup \Omega_r$, then $\omega(X) \in [\low,\upp]$. So assume that it is added in the recursive step. Then by the induction hypothesis, $\omega'(X) \in [(\low+2n)/2,(\upp-2n)/2]$, and $\omega(X) \in [l,u]$ by~\eqref{eq:scale}.
\qed
\end{proof}

\section{Exact Exponential Algorithms} \label{sec:exactalg}
In this section, we demonstrate the applicability of Theorem~\ref{thm:shrinkint} to exact exponential algorithms. First, we consider the relation between the computational complexity of \knapsack and \subsetsum. It is trivial that any algorithm for \knapsack can be used for \subsetsum: Given an instance $(U,\omega,b)$ of \subsetsum, we can define $t = b$ and $\nu(e)=\omega(e)$ for every $e \in U$. Then the instance of \subsetsum is a {\tt Yes}-instance if and only if the constructed instance of \knapsack is a {\tt Yes}-instance. This can be decided by the assumed algorithm for \knapsack. We prove the converse relation below by applying Theorem~\ref{thm:shrinkint}.

\newcommand{\conc}{\mathtt{conc}} 
 
\begin{theorem}\label{thm:sssknp}
If there exists an algorithm that decides the \subsetsum problem in $\OHS(t(n))$ time and $\OHS(s(n))$ space, then there exists an algorithm that decides the \knapsack problem in $\OHS(t(n))$ time and $\OHS(s(n))$ space.
\end{theorem}
\begin{proof}
Consider a \knapsack instance, consisting of a universe $U=\{1,\ldots,n\}$, weight functions $\omega, \nu: U \rightarrow \{1,\ldots, N\}$, and integers $b,t$. Now we apply Theorem~\ref{thm:shrinkint} on both weight functions to obtain two sets $\Omega_{\omega}, \Omega_{\nu}$. By considering the elements of $\Omega := \Omega_{\omega} \times \Omega_{\nu}$ as quadruples, we obtain a set $\Omega$ of at most $\OH(n^2 \lg^2 (nN))$ quadruples $( \omega_i ,\nu_i,b_i,t_i)$ such that for every $X \subseteq U$ it holds that $\nu(X) \in [0,b]$ and $\omega(X) \in \left[t,nN\right]$ if and only if there exists an $i$ such that $\nu_i(X)=b_i$ and $\omega_i(X)=t_i$.
It remains to show that, for every quadruple $(\omega_i,\nu_i,b_i,t_i)$, we can determine whether there exists an $X \subseteq U$ such that $\nu_i(X)=b_i$ and $\omega(X)=t_i$. To do this, we create an instance of \subsetsum by concatenating the integer values. Specifically, given a quadruple $(\omega_i,\nu_i,b_i,t_i)$, define
\[
	\alpha_i(e)= \nu_i(e)N(n+1) +\omega_i(e) \text{ for every } e \in U \qquad \text{and} \qquad  c_i = b_i(n+1)N+t_i.
\]
It is easy to see that, since $\omega_i(X) \leq N(n+1)$, $\alpha_i(X)=c_i$ if and only if $\nu_i(X)=b_i$ and $\omega(X)=t_i$. Then the assumed algorithm for \subsetsum can be used to decide for every quadruple $(\omega_i ,\nu_i,b_i,t_i)$ whether there exists $X \subseteq U$ such that $\alpha_i(X)=c_i$. This in turn enables us to decide the \knapsack instance. The bound on the time and space complexity follows immediately from the fact that $|\Omega|$ is $\OH(n^{2} \lg^2 (nN))$, which is polynomial in the size of the instance. \qed
\end{proof}
As an easy corollary, we observe that we can apply binary search to even deal with the maximization variant of \knapsack.

\begin{corollary}
There exists an algorithm that decides the \subsetsum problem in $\OHS(t(n))$ time and $\OHS(s(n))$ space if and only if there exists an algorithm that solves the {\sc Maximum} \knapsack problem in $\OHS(t(n))$ time and $\OHS(s(n))$ space.
\end{corollary}

We can use the ideas in the proof of Theorem~\ref{thm:sssknp} to give another result on \knapsack. To this end, we assume that the given instance of \knapsack or \subsetsum is sparse, that is, the number of distinct sums in the instance is small. We recall a recent result of Kaski et al.~\cite{Kaski11}.

\begin{theorem}[\cite{Kaski11}]\label{thm:kaski}
There is an algorithm that decides an instance $(U,\omega,t)$ of \subsetsum in $\OHS(S)$ expected time and $\OHS(1)$ space, where $S=|\{ \omega(X): X \subseteq U \}|$.
\end{theorem}
Using Theorem~\ref{thm:shrinkint} and the ideas of Theorem~\ref{thm:sssknp}, we can prove the following.

\begin{theorem}\label{thm:sparse}
There is an algorithm that decides an instance $(U,\omega,\nu,t,b)$ of \knapsack in $\OHS(S)$ expected time and $\OHS(1)$ space, where $S=|\{ (\omega(X),\nu(X)): X \subseteq U \}|$.
\end{theorem}
To prove Theorem~\ref{thm:sparse}, we require the following auxiliary lemma.

\begin{lemma}\label{lem:sums}
Let $U$ be a set of $n$ elements, let $\omega,\nu : U \rightarrow \{1,\ldots,N\}$ be weight functions, let $p\geq1$ be an integer, and let $\est{\omega}(e)=\left\lfloor  \omega(e)/p \right\rfloor$ and $\est{\nu}(e)=\left\lfloor  \nu(e)/p \right\rfloor$ for every $e \in U$. Then
\[
	|\{(\est{\omega}(X),\est{\nu}(X)) : X \subseteq U \}| \leq n^2 \cdot|\{ (\omega(X),\nu(X)) : X \subseteq U \}|.
\]
\end{lemma}
\begin{proof}
For any pair of integers $(x,y)$, consider the set $\mathcal{Z} = \{X \subseteq U : \omega(X)=x, \nu(X)=y \}$. For each $X \in \mathcal{Z}$, we have that $\omega(X)-n \leq  p \cdot \est{\omega}(X) \leq \omega(X)$ and $\nu(X)-n \leq  p \cdot \est{\nu}(X) \leq \nu(X)$. Therefore, $|\{ (\est{\omega}(X), \est{\nu}(X)) : X \in \mathcal{Z} \}| \leq n^2$, and the lemma follows.
\qed \end{proof}

\begin{proof}[of Theorem~\ref{thm:sparse}]
We first apply the same construction as in the proof of Theorem~\ref{thm:sssknp} to obtain pairs $(\alpha_i,c_i)$. We then apply the algorithm of Theorem~\ref{thm:kaski} on all of these pairs and return {\tt Yes} if the algorithm finds an index $i$ and a set $X \subseteq U$ such that $\alpha_i(X)=c_i$.

It remains to prove that this introduces at most a polynomial overhead. Since the number of pairs is bounded by a polynomial in the input length, it suffices to show that for every $i$, the quantity $|\{ \alpha_i(X): X \subseteq U\}|$ is at most $\OHS(S)$. Observe that new weight functions are created in two places. First when Theorem~\ref{thm:shrinkint} is invoked: note that in Algorithm~\ref{alg:shrink}, all created weight functions are effectively obtained by halving the weights $x$ times and rounding down, which is equivalent to truncating the bitstring or dividing by $2^x$ and rounding down. Hence, for all created weight functions $\omega_i$, we have that $|\{(\omega_i(X), \nu_i(X)): X \subseteq U\}| \leq n^2 \cdot S$ for every $i$ by Lemma~\ref{lem:sums}. 
The second place is when $\alpha$ is defined by concatenating the integers: then $|\{\alpha_i(X): X \subseteq U\}|=|\{(\omega_i(X), \nu_i(X)): X \subseteq U\}| \leq n^2 \cdot S$. Hence the overhead is at most polynomial.
\qed \end{proof}

\section{Kernelization} \label{sec:kernel}

In this section we show that Theorem~\ref{thm:shrinkint} can be used in combination with a known kernelization technique to reduce the number of bits needed to represent the weights of weighted minimization problems to an amount that is polynomial in the number of bits needed to represent the remainder of the input instance.

\begin{theorem}\label{thm:kernel}
The weighted variants of the \vc and \ds problems, \tsp, and \knapsack all admit polynomial randomized Turing kernels when parameterized by $|U|$.
\end{theorem}
We need the following lemma from Harnik and Naor~\cite{DBLP:journals/siamcomp/HarnikN10}, which uses randomization to reduce the weights.

\begin{lemma}[\cite{DBLP:journals/siamcomp/HarnikN10}]\label{lem:HarnikNaor}
Let $U$ be a set of size $n$. There exists a polynomial-time algorithm that, given $\omega: U \rightarrow \{0,\ldots,N\}$, an integer $t$, and a real $\epsilon > 0$, returns $\omega': U \rightarrow \{0,\ldots, M\}$ and integers $t_1,\ldots,t_n \leq M$ where $M \leq 2^n\cdot\poly(n,\lg N,\epsilon^{-1})$ such that for every set family $\mathcal{F} \subseteq 2^U$:
	\begin{enumerate}[(R1)]
		\item if there is an $X \in \mathcal{F}$ such that $\omega(X) = t$, then there exist $i$ such that $\omega'(X) = t_i$, 
		\item if there is no $X \in \mathcal{F}$ such that $\omega(X) = t$ then
			\begin{equation}\label{eq:probbound}
				\prob[\text{there exist } i \text{ and } X \in \mathcal{F} \text{ such that } \omega'(X) = t_i] \leq \mathcal{O}(\epsilon).
			\end{equation}
	\end{enumerate}
\end{lemma}
We give the proof of the lemma in the appendix for completeness. It relies on the fact that in every interval of length $l$ the number of primes is roughly $l / \ln l$ and that a random prime can be constructed in time polylogarithmic in the upper bound of the interval.

\begin{proof}[of Theorem~\ref{thm:kernel}]
In all problems mentioned in the statement of Theorem~\ref{thm:kernel}, there is a set family $\mathcal{F} \subseteq 2^U$ and we are asked whether there exists an $X \in \mathcal{F}$ such that either $\omega(X) \in [0,t]$ or $\omega(X) \in [t,nN]$ (depending on the problem). Note that we can assume that $\lg N \leq 2^{|U|}$; otherwise the input is of size at least $2^{|U|}$ and we can use a trivial brute-force algorithm to solve the instance and reduce it to an equivalent instance of constant size. 

Now we use Theorem~\ref{thm:shrinkint} to obtain a set $\Omega$ of $\ell = \mathcal{O}(n \lg (nN))$ pairs $(\omega_i,t_i)$ and reduce the original problem to detecting whether there exists a pair $(\omega_i,t_i) \in \Omega$ and $X \in \mathcal{F}$ such that $\omega_i(X)=t_i$. To reduce the latter problem further, we apply the algorithm of Lemma~\ref{lem:HarnikNaor}, setting $\epsilon = \epsilon'/\ell$ for some small value of $\epsilon'$. Hence, for every $(\omega_i,t_i)$, we obtain a weight function $\omega'_i$ and $n$ integers $t'_{i1},\ldots,t'_{in}$ such that if there exists an $X \in \mathcal{F}$ with $\omega_i(X)=t_i$, then there exists a $j$ such that $\omega'_i(X)=t'_{ij}$ for some $j$ and otherwise~\eqref{eq:probbound} holds. Hence, this procedure generates $\OH(n^2 \lg (nN))$ pairs such that $(i)$ if there is $X \in \mathcal{F}$ with $\omega(X) \in [0,t]$, a pair $(\omega'_i,t'_{ij})$ with $\omega'_i(X)=t'_{ij}$ is generated $(ii)$ if there is no $X \in \mathcal{F}$ with $\omega(X) \in [0,t]$:
\[
	\prob[\text{there exist } i,j \text{ and } X \in \mathcal{F} \text{ such that } \omega'_i(X) = t'_{ij}] \leq \ell \cdot  \epsilon = \OH(\epsilon').
\]

Now we have reduced the original decision problem to a problem that is clearly in NP: indeed, we can obtain the correct $X \subseteq U$ in non-deterministic polynomial time and verify whether it satisfies $X \in \mathcal{F}$ (that is, is it a vertex cover, dominating set, ...) and $\omega'_i(X) = t'_{ij}$. Then, since the original problem (\vc, \ds, ...) is NP-complete, we can reduce the problem to instances of the original problem with a Karp-reduction. Hence we have reduced one problem instance to many problem instances such that:
\begin{itemize}
	\item if the original instance is a {\tt Yes}-instance, then one of the created instances is also a {\tt Yes}-instance;
	\item if the original instance is a {\tt No}-instance, then with constant probability all created instances are {\tt No}-instances.
\end{itemize}

Thus it remains to show that the number and the description lengths of the created instances are bounded by a polynomial in the original input size. To see that this is the case, first note that after applying Lemma~\ref{lem:HarnikNaor} we have $\OH(n^2 \lg (nN))$ pairs of weight functions bounded by $2^n \cdot \poly(n, \lg N, \epsilon^{-1})$. Since we assumed that $\lg N \leq 2^n$, these weight functions are represented by polynomially many bits. Then, the theorem follows from the fact that a Karp-reduction increases the size of a problem by at most a polynomial factor.
\qed \end{proof}

\section{Conclusion} \label{sec:conclusion}
We presented a generic and simple method to convert ranged problems into exact problems. While this result is already interesting by itself given its generality, we also gave a number of corollaries that followed by combining our method with techniques for exact problems already available from previous work. It is worth emphasizing the generality of our results in Section~\ref{sec:maintool} and Section~\ref{sec:exactalg}. For example, in the context of exact exponential algorithms, \tsp seems to be significantly harder than its unweighted version, {\sc Hamiltonian Cycle}. Recently, the latter was shown to be solvable in $\OHS(1.66^{n})$ time and polynomial space~\cite{DBLP:conf/focs/Bjorklund10}, whereas the best algorithm for \tsp that is insensitive to large weights uses $\OHS(2^n)$ and space. By combining the hashing idea of~\cite{Kaski11} with our method, it is for example possible to obtain a polynomial-space algorithm for \tsp that runs in $\OHS(2^n W)$ time, where $W= |\{ \omega(X) : X \text{ is a Hamiltonian cycle.}\}|$.

We leave the reader with several interesting open questions:
\begin{itemize}
\item Can we get a ``classical''(i.e. non-Turing or many-to-one) polynomial kernel for the considered parameterized version of weighted \vc?
\item Further, significant reduction of the weights in polynomial time seems hard (for example, it would imply an improved pseudo-polynomial algorithm for \subsetsum), but is it possible in pseudo-polynomial time for example for \tsp? \item When is minimizing/maximizing as hard as the general range problem?
\item Can we modify Theorem~\ref{thm:shrinkint} to make it counting-preserving? More precisely, can we obtain a variant of the theorem with a third condition $(C3)$ saying that for every $X \subseteq U$ there is at most one $i$ such that $\omega_{i}(X)=t_i$?
\end{itemize}

\bibliographystyle{abbrv}
{\bibliography{IterativeRounding}}
\end{document}